\documentclass[letterpaper, 11pt]{article}

\usepackage{appendix}

\newcommand{\comment}[1]{}

\usepackage{graphicx}
\usepackage{latexsym}
\usepackage{amssymb}
\usepackage{mdwlist}
\usepackage{pxfonts}

\usepackage{algorithm}

\newcommand{\teq}[1]{\stackrel{#1}{=}}
\newcommand{\nteq}[1]{\stackrel{#1}{\neq}}

\newcommand{\nchoosek}[2]{{#1 \choose #2}}

\newcommand{\mybox}[1]{\centerline{\framebox{\parbox[c]{6.5 in}{#1}}}}

\newenvironment{proof}{\paragraph{\bf Proof:}}{\hspace*{\fill}\(\Box\)}
\newtheorem{theorem}{Theorem}

\newtheorem{corollary}{Corollary}

\def\noflash#1{\setbox0=\hbox{#1}\hbox to 1\wd0{\hfill}}

\begin{document}
\title{Capacity of Byzantine Consensus with\\ Capacity-Limited Point-to-Point Links
 \footnote{\normalsize This research is supported
in part by Army Research Office grant W-911-NF-0710287 and National
Science Foundation award 1059540. Any opinions, findings, and conclusions or recommendations expressed here are those of the authors and do not
necessarily reflect the views of the funding agencies or the U.S. government.}}

\author{Guanfeng Liang and Nitin Vaidya\\ \normalsize Department of Electrical and Computer Engineering, and\\ \normalsize Coordinated Science Laboratory\\ \normalsize University of Illinois at Urbana-Champaign\\ \normalsize gliang2@illinois.edu, nhv@illinois.edu}

\maketitle
\date{}


\thispagestyle{empty}

\newpage

\setcounter{page}{1}

\section{Introduction}
\label{sec:intro}

In our previous work \cite{techreport_CBA_final,techreport_Byzantine_complete,infocom-ByzantineBroadcast}, we investigated the {\bf capacity} of the broadcast version of the Byzantine agreement problem \cite{psl82} in networks where communications links are capacity limited. In this report, we are going to study capacity of the {\bf consensus} version of the Byzantine agreement problem. 

The Byzantine consensus problem considers $n$ nodes, namely $P_1,...,P_n$, of which at most $f$ nodes may be {\em faulty} and deviate from the algorithm in arbitrary fashion. Each node $P_i$ is given an input value $v_i$, and they want
to agree on a value $v$ such that the following properties are satisfied:
\begin{itemize}
\item {\em Termination}: every fault-free $P_i$ eventually decides on an output value $v_i'$,
\item {\em Consistency}: the output values of all fault-free nodes are equal, i.e.,
for every fault-free node $P_i$, $v_i'=v'$ for some $v'$,
\item {\em Validity}: if every fault-free $P_i$ holds the same input $v_i=v$ for some $v$, then $v'=v$.
\end{itemize}

\subsection{Models}
In this report, we use the similar network and adversary models as in our previous work \cite{techreport_Byzantine_complete}. 

\subsubsection{Network Model} We assume a synchronous network modeled as a fully-connected directed graph $G(V,E)$, where $V$ is the set of $n$ nodes and $E$ is the set of $n(n-1)$ directed links. Each directed link $e_{ij} = (P_i,P_j)$ is associated with a capacity $c_{ij}$, which specifies the maximum amount of information that can be transmitted on that link per unit time. That is, for any period of duration $t$ and a link $e_{ij}$, up to $t c_{ij}$ bits can be sent from node $P_i$ to $P_j$ on link $e_{ij}$, but no more than $t c_{ij}$. The capacity of some links may be 0, which implies that these links do not exist.


\subsubsection{Adversary Model} We assume that  the adversary has complete knowledge of the network topology, the algorithm, the state of all nodes, the input values, and no secret is hidden from the adversary.
The adversary can take over up to  $f$
processors at any point during the algorithm, where
$f < n/3$. These nodes are said to be {\em faulty}. The faulty nodes can engage in any kind of
deviations from the algorithm, including sending false messages,
collusion, and crash failures. 

\subsection{Capacity of Consensus}
Our goal in this work is to design algorithms that can achieve optimal {\em throughput} of consensus.

When defining throughput, the input value $v_i$ at each node $P_i$ referred in the above definition of
consensus is viewed as an infinite sequence of {\em information} bits. We assume that the information bits have already been compressed, such that for any subsequence of length $l>0$, the $2^l$ possible sequences are sent by the sender with equal probability. Thus, no set of information bits sent by the sender contains useful information about other bits. This assumption comes from the observation about ``typical sequences'' in Shannon's work \cite{shannon}. 

At each node $P_i$, we view
the output value $v_i'$ as being represented in an array of bits.
Initially, none of the bits in this array at a node have been agreed upon.
As time progresses, the array is filled in with agreed bits. In principle, the
array may not necessarily be filled sequentially. For instance, a node may agree
on bit number 3 before it is able to agree on bit number 2. 
Once a node agrees on any bit,
that agreed bit cannot be changed.

We assume that a consensus algorithm begins execution at time 0.
In a given execution of a consensus algorithm,
suppose that by time $t$ all the
fault-free nodes have agreed upon bits 0 through $b(t)-1$, and at least one
fault-free node has not yet agreed on bit number $b(t)$.
Then, the consensus {\em throughput}\, is defined as\footnote{As a technicality,
we assume that the limit exists.}
$\lim_{t\rightarrow \infty} ~ \frac{b(t)}{t}$.

~

\mybox{
{\bf Capacity} of consensus in a given network $G(V,E)$, denoted as $C_{con}(G)$, is defined as the supremum of all achievable
consensus throughputs.
}

\section{Upper Bound}
In this section, we prove an upper bound of the consensus capacity of any network $G(V,E)$. Let us first define some notations. For any subset of nodes  $S\subset V$ such that $|S|\le f$, denote $\Gamma_S = \{\gamma: \gamma\subset V\backslash S ~\mathrm{ and } ~|\gamma| = n-|S|-f\}$, i.e., $\Gamma_S$ is the set of the $\nchoosek{n-|S|}{n-|S|-f}$ subsets of $n-|S|-f$ nodes that do not include nodes in $S$. For every $\gamma \in \Gamma_S$, denote $I_S(\gamma) = \sum_{P_j \in \gamma,P_i\in S}c_{j,i}$ as the incoming capacity to the set of nodes in $S$ from nodes in $\gamma$, and $I_S^* = \min_{\gamma \in \Gamma_S}I_S(\gamma)$ be the minimum over all $\gamma \in \Gamma_S$. 

We prove the following upper bound of the consensus capacity:
\begin{theorem}
\label{thm:upper_bound}
For any network $G(V,E)$, the capacity of consensus satisfies
\begin{equation}
C_{con}(G)\le I^* \triangleq \min_{S\subset V,|S|\le f} I_S^*.
\end{equation}
\end{theorem}

\begin{proof}
Suppose on contrary that $C_{con}(G)>I^*$. Then there must exist a consensus algorithm $A$ that achieves consensus on $b(t) > t I^* $ bits during some period $[0,t]$, with no more than $t c_{ij}$ bits transmitted over link $e_{ij}$ for all $P_i,P_j\in V$. 

\begin{figure}[ht]
\centering
\includegraphics[width = 4 in]{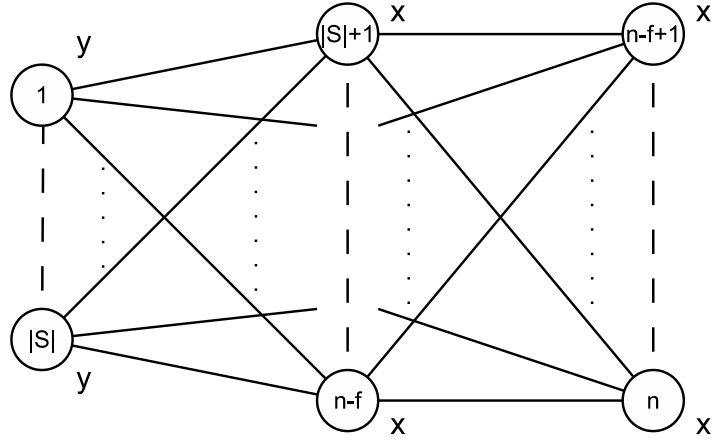}
\caption{State machine for the proof of Theorem \ref{thm:upper_bound}.}
\label{fig:state_machine}
\end{figure}

Without loss of generality, assume that $I^* = I_S(\gamma)$ for $S = \{P_1,\cdots, P_{|S|}\}$, and $\gamma = \{P_{|S|+1},\cdots,P_{n-f}\}$. Now consider the state machine $G'(V',E')$ illustrated in Figure \ref{fig:state_machine}. $G'(V',E')$ is a subgraph of $G(V,E)$, with node $i\in V'$ representing $P_i\in V$, and links between nodes in  $\{1,\cdots, |S|\}$ and nodes $\{n-f+1,\cdots, n\}$ removed from $E$.
The remaining links have the same capacity as in $G(V,E)$. So nodes in $\{1,\cdots, |S|\}$ receive nothing from nodes $\{n-f+1,\cdots, n\}$, and nodes $\{n-f+1,\cdots, P_n\}$ receive nothing from nodes in $\{1,\cdots, |S|\}$. In this state machine, every node $i$ runs the correct code that node $P_i$ in network $G(V,E)$ should run according to algorithm $A$. Nodes in $\{1,\cdots, |S|\}$ are all given the same input value $y$ of $b(t)$ bits, and all the other nodes are given the same input value $x$ of $b(t)$ bits.

Now consider the following two scenarios in $G(V,E)$:
\begin{enumerate}
\item Nodes in $S$ are fault-free and given input value $y$; $P_{|S|+1},\cdots, P_{n-f}$ are fault-free and given input value $x$; Only $P_{n-f+1},\cdots,P_n$ are faulty. The $f$ faulty nodes behave as nodes $n-f+1,\cdots, n$ in $G'(V',E')$ by not sending anything to nodes in $S$ and ignoring any incoming information received from nodes in $S$.
\item $P_{|S|+1},\cdots, P_n$ are fault-free and given input value $x$. Only the $|S| \le f$ nodes in  $S$ are faulty and behave as nodes $1,\dots, |S|$ in $G'(V',E')$ by not sending anything to $P_{n-f+1},\cdots,P_n$ and ignoring any incoming information received from $P_{n-f+1},\cdots,P_n$. According to the validity property, the output value of $P_{|S|+1},\cdots, P_n$ is $x$.
\end{enumerate}
Given that algorithm $A$ solves consensus of $b(t)$ bits by time $t$, the nodes will reach consensus correct by time $t$ in both scenarios.
It is not hard to see that the information observed by $P_{|S|+1},\cdots, P_{n-f}$ in both scenarios is the same as the information observed by node $|S|+1,\cdots, n-f$ in $G'(V',E')$. As a result, they cannot distinguish between the two scenarios and must decide on the same value. So in scenario 1,  $P_{|S|+1},\cdots, P_n$ decide on the output value $x$. Then according to the consistency property, nodes in $S$ will also decide on the same output value $x$ -- the input value of $P_{|S|+1},\cdots, P_{n-f}$.

Now let us fix $y$ in scenario 1 and vary $x$. Notice that in scenario 1 nodes in $S$ receive no more than $t I^* $ bits. Since we assume that $b(t) > t I^*$, according to the pigeonhole principle, there must be two values $x_1\neq x_2$ such that the resulting $t I^*$ bits received by nodes in $S$ are identical. As a result, nodes in $S$ must decide on the same output value when $P_{|S|+1},\cdots, P_{n-f}$ are given the input value $x_1$ or $x_2$. However, as we discussed above,   output value of nodes in $S$  equals to the input value of $P_{|S|+1},\cdots, P_{n-f}$. So nodes in $S$ cannot decide on the same output value when  $P_{|S|+1},\cdots, P_{n-f}$ are given two different input values $x_1$ and $x_2$, which leads to a contradiction.
\end{proof}

\section{Capacity of Complete 4-Node Networks}
\label{sec:4node}

In this section, we prove by construction that the upper bound from Theorem \ref{thm:upper_bound} is tight for complete 4-node networks: 4-node networks in which all directed links have capacity $> 0$.

For the clarity of presentation, we use notations slightly different from the previous sections. We rename the four nodes as A, B, C and D, instead of $P_1,\cdots, P_4$. We denote by XY the directed link from node X to node Y. When it is clear from context, we use $XY$ to also represent the capacity of the directed link XY. We use the notation $\widehat{XY}$ to represent the pair of directed links XY and YX, as well as the sum capacity of this pair of links, i.e., $\widehat{XY}=XY+YX$.

For a 4-node network $G$ with at most 1 faulty node, Theorem \ref{thm:upper_bound} states that $C_{con}(G)$ is no more than the minimum of the sum capacity of any two incoming links to a node, i.e.,
\begin{eqnarray}
C_{con}(G) \le I^* \!\!\!\!&=\!\!\!& \min\{BA+CA,BA+DA,CA+DA,AB+CB,AB+DB,CB+DB,\nonumber \\
&& AC+BC,AC+DC,BC+DC,AD+BD,AD+CD,BD+CD\} \label{eq:4node:bound}.
\end{eqnarray}

We are going to present an algorithm for complete 4-node networks that achieve consensus throughput arbitrarily close to the upper bound specified by Theorem \ref{thm:upper_bound}, i.e., it can achieve consensus throughput of $R$ bits per unit time for all $R<I^*$.
Before we present our capacity achieving algorithm, we first prove the following theorem.

\begin{theorem}
For any positive value $R< I^*$, at least three of $\widehat{AB},\widehat{AC},\widehat{AD},\widehat{BC},\widehat{BD},\widehat{CD}$ are $> R$.
\end{theorem}
\begin{proof}
Consider any subset of three nodes, say $\{A, B, C\}$. We have
\begin{eqnarray}
\widehat{AB} + \widehat{AC} + \widehat{BC} &=& (AB+BA) + (AC+CA) + (BC+CB)\\
&=& (BA+CA) + (AB+CB) + (AC+BC)\\
&>& 3R.
\end{eqnarray}
The inequality follows the fact that $BA+CA,AB+CB,AC+BC$ are all $\ge I^*>R$. It then follows that at least one of $\widehat{AB}, \widehat{AC}, \widehat{BC}$ is $>R$.

Without loss of generality, assume that $\widehat{AB} >R$. For the two subsets of nodes $\{A, C, D\}$ and $\{B, C, D\}$, according to the same argument, at least one of $\{\widehat{AC},\widehat{AD},\widehat{CD}\}$ and one of $\{\widehat{BC},\widehat{BD},\widehat{CD}\}$ are $>R$. There are two cases:
\begin{itemize}
\item $\widehat{CD}\le R$: It follows that one of $\{\widehat{AC},\widehat{AD}\}$ and one of $\{\widehat{BC},\widehat{BD}\}$ are $>R$. So at least three pairs of links that are $>R$.
\item $\widehat{CD}>R$: In this case, we have
\begin{eqnarray}
\widehat{AC} + \widehat{AD} + \widehat{BC} + \widehat{BD}
&=& (AC+CA) + (AD+DA) + (BC+CB) + (BD+DB)\\
&=& (CA+DA) + (CB+DB) + (AC+BC) + (AD+BD)\\
&>& 4R.
\end{eqnarray}
So at least one of $\widehat{AC} , \widehat{AD} , \widehat{BC} , \widehat{BD}$ is $>R$. Again, we have three pairs of links that are $>R$.
\end{itemize}
This proves the theorem.

\end{proof}

Then the corollary below immediately follows.
\begin{corollary}
There must be a subset of three nodes, say \{X,Y,Z\}, such that $\widehat{XY},\widehat{XZ} >R$.
\end{corollary}

Without loss of generality, we can assume that $\widehat{AB}$ and $\widehat{BC}$ are $>R$. Now we can describe our consensus algorithm.  The discussion in this section is not self-contained, and relies heavily on the material in \cite{techreport_Byzantine_complete,infocom-ByzantineBroadcast} -- please refer to \cite{techreport_Byzantine_complete,infocom-ByzantineBroadcast} for the necessary background. 

In our consensus algorithm, the sequence of information bits at each node are divided into generations of $Rc$ bits (the choice of $c$ will be elaborated later), and consensus is achieved one generation after another using the proposed algorithm. Denote $X(g)$ as the input value at node X for the $g$-th generation. $X(g)$ is represented as a vector of $R$ packets, each of which consists of $c$ bits. Similar to \cite{techreport_Byzantine_complete,infocom-ByzantineBroadcast}, we require node X to generate coded packets from $X(g)$ such that every subset of $R$ coded packets consist linearly independent combinations of the $R$ original packets.

The proposed consensus algorithm has five
modes of operation, named {\em Undetected 2=, Undetected 1=1$\neq$, Undetected 2$\neq$, Detected,} and {\em Identified}. The Undetected modes operate when no failure has been detected yet. The Detected mode operates after failure is detected and the location of the faulty node has been narrowed down to a set of 2 nodes. The Identified mode operates after the faulty node has been located.
Similar to our broadcast algorithm in \cite{techreport_Byzantine_complete,infocom-ByzantineBroadcast}, 
repeated (and pipelined) execution of our algorithm
will be used to achieve throughput approaching the capacity.
At time 0 (the first generation),
the network starts in mode Undetected 2=.

Before describing the operations in different modes, we introduce some terminologies to simplify the description.

\begin{enumerate}
\item ``Node X checks the consistency of the packets it has received'':

Node X tries to find the unique solution for every subset of $R$ coded packets it has received. If every such subset has a unique solution, and the solutions of all such subsets are identical, then node X determines these packets to be {\em consistent}, otherwise they are {\em inconsistent}. Node X reliably broadcasts a 1-bit notification indicating whether it finds the received packets consistent or not, using an existing error-free Byzantine broadcast algorithm (e.g. \cite{bit_optimal_89,opt_bit_Welch92}). 

\item If $\widehat{XY} \ge R$, ``the pair of nodes (X,Y) check directly'': 

Nodes X and Y each generates $XY$ and $YX$ coded packets from $X(g)$ and $Y(g)$, and exchange these packets over links XY and YX. Then node X checks if the packets received on link YX are consistent with its own input value $X(g)$, if yes, it reliably broadcasts (using an existing error-free Byzantine broadcast algorithm) a 1-bit notification ``='' or ``$\neq$'' indicating whether the packets received from Y are consistent with $X(g)$ or not. Similar action is taken by node Y. If both notifications from X and Y are ``='', then we say $X(g)\equiv Y(g)$; otherwise we say $X(g)\nequiv Y(g)$. 

Since $\widehat{XY} > R$, it can be concluded that if both nodes X and Y are fault-free, $X(g)\equiv Y(g)$ if and only if $X(g) = Y(g)$, i.e., the result of the checking reflects the relationship between $X(g)$ and $Y(g)$ correctly.

\item ``The pair of nodes (X,Y) check  through node Z'': 

Nodes X and Y each generates $XY,XZ$ and $YX,YZ$ coded packets from $X(g)$ and $Y(g)$, and send these packets over links XY, XZ, YX and YZ. Node Z forwards as many coded packets received from node X as possible  to node Y, and forwards as many packets received from node Y as possible  to node X. Then each of nodes X, Y and Z first checks the consistency of the packets  have received from the other two nodes, according to the description in point 1. If any one of the three nodes finds the received packets inconsistent, failure is detected. Otherwise, node X then checks if the unique solution of the received packets equals to its own input value $X(g)$. Then it broadcasts a 1-bit notification ``='' or ``$\neq$'' indicating whether the packets received from Y and Z are consistent with $X(g)$ or not. Similar action are taken at node Y. If both notifications from nodes X and Y are ``='', we say $X(g)\teq{Z}Y(g)$; otherwise we say $X(g)\nteq{Z}Y(g)$.

It is not hard to show that 
\begin{eqnarray}
&&XY + YX + \min\{XZ,ZY\} + \min\{YZ,ZX\}\\
 &&= XY + YX + \min\{XZ+YZ, XZ+ZX, ZY+YZ, ZY+ZX\}\\
&&=\min\{XY+YX + (XZ+YZ),~ XY+XZ+(YX+ZX),\\
&&~~~~~~~~~~~~YX+YZ+(XY+ZY),~(YX+ZX)+(XY+ZY)\}\\
&&\ge R.
\end{eqnarray}
The inequality follows from the fact that $(XZ+YZ),~(YX+ZX)$ and $(XY+ZY)$ all $>R$.

So if all three nodes X, Y and Z are fault-free, $X(g)\teq{Z} Y(g)$ if and only if $X(g)=Y(g)$, i.e., the result of the checking through node Z reflects the relationship between $X(g)$ and $Y(g)$ correctly.

\item ``The pair of nodes (X,Y) check value $X'(g)$ and $Y(g)$ through node Z'':

As we will see later, this operation is performed only when it has been determined that the input of node X (or Y) is different from the inputs of other nodes. Then all node X has to do is to make sure the packets it receives are consistent with the other nodes. In this case, node X first solves the packets it has received for a unique solution $X'(g)$. Then it checks $X'(g)$ with node Y through node Z according to the operation in point 3, as if $X'(g)$ is its own input.
\end{enumerate}

\subsection{Mode Undetected 2=}
\label{mode:undetected_2=}
For the $g$-th generation, the algorithm operates in this mode if no failure has been detected so far, and the inputs at nodes A, B and C appear to be identical so far, i.e., $A(h)\equiv B(h)$ and $B(h) \equiv C(h)$ for all $h<g$. The first generation operates in this mode too. Mode Undetected 2= proceeds in rounds as described below:

\begin{enumerate}
\item The two pairs of nodes (A,B) and (B,C) check directly.

\item If [$A(g)\equiv B(g)$ and $B(g)\nequiv C(g)$] or [$A(g)\nequiv B(g)$ and  $B(g)\equiv C(g)$]: 

The algorithm aborts the current generation, switches to mode Undetected 1=1$\neq$, and restarts the current generation with the new mode. Following generations will also operate in the new mode.

\item  If $A(g)\nequiv B(g)$ and $B(g)\nequiv C(g)$:

The algorithm aborts the current generation, switches to mode Undetected 2$\neq$, and restarts the current generation with the new mode. Following generations will also operate in the new mode.

\item If $A(g)\equiv B(g), B(g)\equiv C(g)$:

The pair of nodes (A,C) check through node D.
\begin{enumerate}
\item If $A(g)\nteq{D} C(g)$: 

It contradicts with the condition $A(g)\equiv B(g), B(g)\equiv C(g)$. So failure is detected. A full-broadcast similar to the extended round 3 in \cite{infocom-ByzantineBroadcast} is performed to diagnose the system. The location of the faulty node will be either narrowed down to a subset of 2 nodes, or correctly identified by the fault-free nodes. Then the algorithm switches to mode Detected or Identified, for the following generations. In particular, if the faulty node is narrowed down to a subset of 2 nodes, then switch to mode Detected; If the faulty node has been identified, then switch to mode Identified.

\item If $A(g)\teq{D} C(g)$:

Node B  sends $BD$ coded packets generated from  $B(g)$ to node D on link BD. Node D checks the consistency of the packets it has received from nodes A, B and C. 

\begin{enumerate}
\item \label{step:2=:decide}
If node D finds the packets consistent, it broadcasts this finding and:

 Nodes A, B and C decide on the output of the current generation as $A(g)$, $B(g)$ and $C(g)$, respectively. Node D decides on the unique solution of the received packets. 

\item Otherwise, failure is detected.  Then the full-broadcast is performed to diagnose the failure, and the system switches to mode Detected or Identified according to the outcome of the diagnosis.
\end{enumerate}
\end{enumerate}
\end{enumerate}

\begin{theorem}[Correctness of Mode Undetected 2=]
\label{thm:2=}
If the fault-free nodes decide on an output for the $g$-th generation in mode Undetected 2=, then the decided outputs are identical and equal to the input of this generation at the fault-free nodes in set \{A,B,C\}.
\end{theorem}
\begin{proof}
As we have seen, in mode Undetected 2=, the nodes can  decide on an output only if $A(g)\equiv B(g)$, $B(g)\equiv C(g)$ and $A(g)\teq{D} C(g)$. It is not hard to see that, at least one of the sets \{A,B\}, \{B,C\}, and \{A,C,D\} consists of only fault-free nodes. Then no matter which node is faulty, there are at least 2 fault-free nodes in the set \{A,B,C\}, and they must have identical input values for the $g$-th generation. So in Step \ref{step:2=:decide}, the fault-free nodes in \{A,B,C\} decide on identical outputs. 

In addition, in the case node D is fault-free, it receives at least $R$ coded packets from the fault-free nodes in the set \{A,B,C\} by \ref{eq:4node:bound} and the fact that $R<I^*$, of which the unique solution is the input of the fault-free nodes in \{A,B,C\}. As a result, if node D does not detect a failure, its output value of the $g$-th generation equals to the input of the other fault-free nodes.
\end{proof}

\subsection{Mode Undetected 1=1$\neq$}
\label{mode:undetected_1=}
For the $g$-th generation, the algorithm operates in this mode if no failure has been detected, and either (1) $A(h)\equiv B(h)$ for all $h\le g$ and $B(h)\nequiv C(h)$ for some $h\le g$; or (2) $B(h) \equiv C(h)$ for all $h\le g$ and $A(h)\nequiv B(h)$ for some $h\le g$. 

Without loss of generality, let us assume $A(h)\equiv B(h)$ for all $h\le g$ and $B(h)\nequiv C(h)$ for some $h\le g$.
Mode  Undetected 1=1$\neq$ proceeds as follows:
\begin{enumerate}
\item The pair of nodes (A,B) check directly. 

\item If $A(g)\nequiv B(g)$:

The algorithm aborts the current generation, switches to mode Undetected 2$\neq$, and restarts the current generation with the new mode. Following generations will also operate on the new mode.

\item If $A(g)\equiv B(g)$:

The pair of nodes (A,B) check through node C. Notice that no more communication on over links $\widehat{AB}$ is needed since they can use the packets they have exchanged directly in step 1 for the purpose of checking through node C.
\begin{enumerate}
\item If $A(g)\nteq{C} B(g)$:

It contradicts with the condition $A(g)\equiv B(g)$. So failure is detected.   Then the full-broadcast is performed to diagnose the failure, and the system switches to mode Detected or Identified according to the outcome of the diagnosis.

\item If $A(g)\teq{C} B(g)$:

Denote $C'(g)$ as the unique solution of the packets node C has received from nodes A and B. 
As we will show in the proof of Theorem \ref{thm:1=}, $C'(g)\equiv B(g)$. 
Then the pair of nodes (A,C) check $A(g)$ and $C'(g)$ through node D. 

\begin{enumerate}
\item If $A(g)\nteq{D} C'(g)$:

It contradicts with the condition $A(g)\equiv B(g)$ and $C'(g)\equiv B(g)$. So failure is detected.  Then the full-broadcast is performed to diagnose the failure, and the system switches to mode Detected or Identified according to the outcome of the diagnosis.

\item If $A(g)\teq{D} C'(g)$: 

Node B sends $BD$ coded packets generated from $B(g)$ to node D. Node D then checks the consistency of the packets it has received.

\begin{enumerate}
\item If node D finds the packets consistent:

Nodes A, B and C decide on $A(g)$, $B(g)$ and $C'(g)$ respectively. Node  D decides on the unique solution of the received packets.

\item Otherwise, failure is detected.  Then the full-broadcast is performed to diagnose the failure, and the system switches to mode Detected or Identified according to the outcome of the diagnosis.

\end{enumerate}
\end{enumerate}
\end{enumerate}
\end{enumerate}

\begin{theorem}[Correctness of Mode Undetected 1=1$\neq$]
\label{thm:1=}
If the nodes decide on an output for the $g$-th generation in mode Undetected 1=1$\neq$, then the decided outputs are identical and equal to the input of this generation at the fault-free node(s) in set \{A,B\} .
\end{theorem}

\begin{proof}

In mode Undetected 1=1$\neq$, the nodes can  decide on an output only if $A(g)\equiv B(g)$, $A(g)\teq{C} B(g)$ and $A(g)\teq{D} C'(g)$. 
Notice that the number of coded packets exchanged between nodes B and C in step 3 is
\begin{equation}
BC + \min\{AC,CB\} =\min \{BC+AC,BC+CB\} > R.
\end{equation}
The inequality is due to $BC+AC>R$ and $BC+CB>R$. So $B(g)\equiv C'(g)$. Similar to the proof of Theorem \ref{thm:2=}, after nodes (A,C) check $A(g)$ and $C'(g)$ through node D, it can be sure that the values $A(g),B(g),C'(g)$ are identical at the fault-free nodes in set \{A,B,C\}. Then the rest follows the same proof of Theorem \ref{thm:2=}.
\end{proof}

\subsection{Mode Undetected 2$\neq$}
\label{mode:undetected_2neq}
For the $g$-th generation, the algorithm operates in this mode if no failure has been detected,  $A(h)\nequiv B(h)$ for some $h\le g$, and $B(h) \nequiv C(h)$ for some $h\le g$.  Mode Undetected 2$\neq$ proceeds as follows:

\begin{enumerate}
\item The pair of nodes (A,C) check through node B and node D.

\item If $A(g)\nteq{B} C(g)$ and $A(g)\nteq{D} C(g)$:

In this case, if nodes A and C are both fault-free, then $A(g)\neq C(g)$ must be true. Otherwise, if node A is faulty, then nodes B and C are both fault-free and $C(h)\neq B(h)$ for some $h\le g$; Similar if node C is faulty. So no matter which node is faulty, it can now be concluded with certainty that the fault-free nodes do not all have identical input value. So the algorithm can decide on a default value and terminate.

\item If $A(g)\teq{B} C(g)$ and $A(g)\nteq{D} C(g)$; or $A(g)\nteq{B} C(g)$ and $A(g)\teq{D} C(g)$:

Failure is detected.  Then the full-broadcast is performed to diagnose the failure, and the system switches to mode Detected or Identified according to the outcome of the diagnosis.

\item If $A(g)\teq{B} C(g)$ and $A(g)\teq{D} C(g)$:
Node B forwards as many packets received from nodes A and C as possible to node D on link BD.
Node D forwards as many packets received from nodes A and C as possible to node B on link DB. Then nodes B and D check the consistency of the packets they have received.

\begin{enumerate}
\item If both nodes B and D find the received packets consistent:

 Nodes A and C decide on $A(g)$ and $C(g)$ respectively. Nodes B and D decide on the unique solution of the received packets.

\item Otherwise, failure is detected. Then the full-broadcast is performed to diagnose the failure, and the system switches to mode Detected or Identified according to the outcome of the diagnosis.
\end{enumerate}
\end{enumerate}

\begin{theorem}[Correctness of Mode Undetected 2$\neq$]
\label{thm:0=}
If the nodes decide on an output for the $g$-th generation in mode Undetected 2$\neq$, then the decided outputs are identical and equal to the input of this generation at the fault-free node(s) in set \{A,C\} .
\end{theorem}

\begin{proof}
In mode Undetected 2$\neq$, the nodes can  decide on an output only if $A(g)\teq{B} C(g)$ and $A(g)\teq{D} B(g)$. 
Similar to the proof of Theorem \ref{thm:1=}, when nodes (A,C) check through node B, node B exchanges $>R$ packets with each one of nodes A and C. If we denote $B'(g)$ as the unique solution of the packets node B has received, it follows $A(g)\equiv B'(g)$ and $B'(g)\equiv C(g)$. Then the rest follows the same proof of Theorem \ref{thm:2=}.
\end{proof}

\subsection{Mode Detected}
\label{mode:detected}
The algorithm operates in this mode if failure has been detected and the location of the faulty node is narrowed down to a subset of two nodes. Without loss of generality, assume that the faulty node has been narrowed down to the set $\{B,D\}$. It is worth noting that in this case, all fault-free nodes know that nodes A and C must be fault-free.

\begin{enumerate}
\item The pair of nodes (A,C) check through node B and node D.

\item If $A(g)\nteq{B} C(g)$ and $A(g)\nteq{D} C(g)$:

In this case, fault-free nodes A and C must have different input values. So the algorithm can decide on a default value and terminate.

\item If $A(g)\teq{B} C(g)$ and $A(g)\nteq{D} C(g)$; or $A(g)\nteq{B} C(g)$ and $A(g)\teq{D} C(g)$:

Failure is detected.  Then the full-broadcast is performed to diagnose the failure, and the system switches to mode Identified according to the outcome of the diagnosis.

\item If $A(g)\teq{B} C(g)$ and $A(g)\teq{D} C(g)$:

Nodes A and C decide on $A(g)$ and $C(g)$ respectively. Nodes B and D decide on the unique solution of the  packets received from nodes A and C. The fault-free node of \{B,D\} has received enough ($R$) packets since $AB+CB,AD+CD>R$.
\end{enumerate}

The proof of correctness of this mode is trivial. So we do not include it in this report.

\subsection{Mode Identified}
\label{mode:identified}
The algorithm operates in this mode if failure has been detected and the location of the faulty node is narrowed down to a specific node. Without loss of generality, assume that the faulty node has been identified as node $B$. In this case, nodes A, C and D are sure that they are fault-free.

\begin{enumerate}
\item The pair of nodes (A,C) check through node D.

\item If $A(g)\nteq{D} C(g)$:

In this case, fault-free nodes A and C must have different input values. So the algorithm can decide on a default value and terminate.

\item If $A(g)\teq{D} C(g)$:

Nodes A and C decide on $A(g)$ and $C(g)$ respectively. Node D decides on the unique solution of the  packets received from nodes A and C.
\end{enumerate}

The proof of correctness of this mode is trivial. So we do not include it in this report.

\subsection{Throughput Analysis}
\label{subsec:throughput}
The throughput analysis of the proposed Byzantine consensus algorithm for complete 4-node networks are similar to the one for our Byzantine broadcast algorithms in \cite{techreport_Byzantine_complete,infocom-ByzantineBroadcast}. 
It is not hard to see that, in every generation, the usage of each link for 
communications of {\em data} packets is within the limit of link capacity. Communication overhead includes: reliable broadcast of 1-bit notifications, dropping existing generations when the algorithm switches between modes, and the diagnosis process after failure is detected. 

Similar to \cite{infocom-ByzantineBroadcast}, each reliable broadcast of 1-bit notifications uses a constant number of bits on each link, independent of $Rc$ -- the number of bits to agree per generation. Since there are a constant number of 1-bit notifications per generation, the reduction in capacity from the notification overhead can be made arbitrarily close to 0 by increasing $c$. Besides, there will be at most 4 mode transitions and 2 diagnosis process throughput the infinite time horizon. So the reduction in capacity from mode transitions and full-broadcast diminishes to 0 as time goes to $\infty$.

Similar to \cite{infocom-ByzantineBroadcast}, the proposed consensus algorithm can achieve any throughput $R < I^*$ in complete 4-node networks by pipelining. Together with Theorem \ref{thm:upper_bound}, we then can conclude that $C_{con} = I^*$ for complete 4-node networks. 



\bibliography{../PaperList}


\end{document}